\documentclass[12pt,a4paper,reqno,twoside]{amsart}

\usepackage{epsfig}
\usepackage{xspace}
\usepackage{graphicx}
\usepackage{amssymb}
\usepackage{algorithm}
\usepackage[noend]{algpseudocode}
\usepackage{amsmath}
\usepackage{amsthm}
\usepackage{amsfonts}
\usepackage{setspace}

\newtheorem{theorem}{Theorem}[section]
\newtheorem{lemma}[theorem]{Lemma}

\newtheorem{definition}[theorem]{Definition}

\newcommand{\R}{{\mathbb R}}

\newcommand{\ie}{{\em i.e.},\xspace }
\newcommand{\etal}{{\em et al.}\ }

\begin{document}

\title{Analytic Connectivity in General Hypergraphs}  
\author[]{Ashwin Guha}
\author[]{Muni Sreenivas Pydi}
\author[]{Biswajit Paria}
\author[]{Ambedkar Dukkipati}
\address{Department of Computer Science and Automation \\ Indian Institute of
Science \\ Bangalore 560012, India.}
\email{ad@csa.iisc.ernet.in}

\begin{abstract}
In this paper we extend the known results of analytic connectivity 
to non-uniform hypergraphs. We prove a modified Cheeger's inequality and also give a bound on analytic connectivity 
with respect to the degree sequence and diameter of a hypergraph.  
\end{abstract}

\maketitle

\section{Introduction}
\label{intro}
Combinatorial graph theory is a vast active area of research with a wide variety of applications. It is well-known that many interesting problems have no polynomial-time algorithms. In an attempt to find approximate solutions to these problems, matrices pertaining  to the graph such as adjacency matrix or Laplacian matrix were studied, and the eigenvalues of these matrices were used to bound various graph parameters. This gave rise to spectral graph theory, which has since become a separate area of research in its own right.

The idea of graphs has been generalized to hypergraphs where an edge may span more than two vertices. Hypergraphs have also been 
explored in detail \cite{berge1973graphs}. However, spectral methods for hypergraphs have received less attention. Recently there 
has been a renewed interest in spectral hypergraph theory. 

 The traditional approach towards dealing with hypergraphs has been to convert it into a graph and apply the known graph results. 
 Another approach is to define a tensor for a hypergraph that is an analogue of the corresponding matrix for a graph. A novel 
 definition of eigenvalue for a tensor given by Qi \cite{qi2005eigenvalues} and independently by Lim \cite{lim2006} has led to surge of activity in this area. Subsequently various results on spectra of hypergraphs  have been provided in \cite{qi2014hvaluesoflaplacian, cooper2012spectra} . The results in these works are applicable to uniform hypergraphs where the size of every edge is fixed. There have been recent attempts to extend these results for non-uniform hypergraphs \cite{banerjee2016spectra, bu2016spectral}.

The concept of analytic connectivity was introduced in \cite{qi2014hvaluesoflaplacian}, which is the analogue of algebraic connectivity introduced by Fiedler \cite{fiedler1973algebraic} for graphs. In \cite{li2016analytic}, the relations between analytic connectivity and other graph parameters are explored. In particular, Cheeger's inequality for uniform hypergraphs is proved.

In this paper we use the terminology given by Banerjee \etal \cite{banerjee2016spectra} and extend the results by Li \etal \cite{li2016analytic} to non-uniform hypergraphs. We prove a modified Cheeger's inequality and also give the relation between analytic connectivity and diameter.

The paper is organized as follows. In Section \ref{prelims} we introduce our notation and state the preliminaries. We state the equivalent theorems for graphs and uniform hypergraphs. In Section \ref{mainres}, we provide the main results with proof. Section \ref{summary} contains concluding remarks.

\section{Preliminaries}
\label{prelims}


Let $H(V,E)$ be a general hypergraph on vertex set $V$ and edge set $E$ which is a set of non-empty subsets of $V$. 
Let $n$ be number of vertices. We denote  by $m$ and $s_{\text{min}} $ the size of the largest edge and 
the smallest edge respectively, \ie $m= \max \{|e|: e \in E \} $ and $s_{\text{min}}= \min \{|e|: e \in E \} $. Let 
$\Delta $ be the largest degree among vertices of $H$.

For a hypergraph $H$ we adopt the definition given in \cite{banerjee2016spectra} and define the adjacency tensor 
$\mathcal{A} $ to be an $n$-dimensional hypermatrix 
\footnote{In this paper we use the terms hypermatrix and tensor interchangeably for ease of understanding.} 
of order $m$ such that if $e = \{v_{i_1} , v_{i_2} ,\ldots, v_{i_s} \} \in E$ is an 
edge of cardinality $s \leq m$, then 
\[  a_{ p_1p_2 \ldots p_m} =  \frac{s}{\Omega},\text{ where } \Omega = \displaystyle\sum_{k_j \geq 1} \frac{m!}{k_1 !k_2 !\ldots k_s!}\] 
with $\sum k_j =m $  and $p_1 , p_2 , \ldots, p_m$ are chosen in all possible ways from $\{i_1 ,i_2 ,\ldots ,i_s \}$ such that 
each element of the set appears at least once. The other entries of the hypermatrix are zero. Note that this definition agrees with the 
definition of adjacency tensor for uniform hypergraphs given in \cite{cooper2012spectra}.

For an edge $e = \{v_{i_1} , v_{i_2} ,\ldots, v_{i_s} \} \in E$ we define
$$x^e_m = \sum x_{r_1} x_{r_2} \ldots x_{r_m},$$  
where the sum is over $ r_1, r_2, \ldots,r_m$  chosen in all possible ways from $\{ i_1 ,i_2 ,\ldots ,i_s\}$ 
with each element of the set appearing at least once in the index .

Let $\mathcal{D}$ be the degree tensor which is a diagonal tensor of order $m$ and dimension $n$ such that $d_{ii\ldots i}= d(v_i)$ and zero 
elsewhere. We define Laplacian tensor $\mathcal{L}(H)$ to be $\mathcal{D} - \mathcal{A}$. 

For a tensor $\mathcal{T}$ of order $m$ and dimension $n$ and $x \in \R^n$ we define $\mathcal{T}x^m$ as 
\[(\mathcal{T}x^m)_j = \sum_{i_2,\ldots, i_m =1 }^{n} t_{ji_2 i_3 \ldots i_m}x_{i_2} x_{i_3}\ldots x_{i_m}.  \]

In particular we have
\begin{align*}
\mathcal{L}x^m &= (\mathcal{D} - \mathcal{A})x^m =  \sum_{i=1}^{n} d(v_i)x_i^m - \sum_{i_2, \ldots, i_m =1}^{n}a_{i i_2 \ldots i_m} x_{i_2}\ldots x_{i_m} \\
& = \sum_{e=\{v_{i_1, \ldots, v_{i_s}} \} \in E} \left( \sum_{j=1}^{s} x_{i_j}^m  - \frac{s}{\Omega} x^e_m \right)
 = \sum_{e \in E}\mathcal{L}(e)x^m
\end{align*}

where $\mathcal{L}(e)x^m = \displaystyle\sum\limits_{j=1}^{s} x_{i_j}^m  - \frac{s}{\Omega} x^e_m $.

It can be shown that for $x \in \R^n_+$, $\mathcal{L}(e)x^m \geq 0 $ \cite{banerjee2016spectra}. We recall the definitions of a few hypergraph parameters which are intuitive generalizations of those for a 
$2$-graph. 

\begin{definition}
For a hypergraph $H$, the isoperimetric number $i(H)$ is defined as 
$$i(H) = \min \left\{ \frac{ |\partial S|}{|S|} : S \subset V, 0 \leq |S| \leq \frac{|V|}{2} \right\}, $$
\end{definition}
\noindent where $\partial S$ is the boundary of $S$ which consists of the edges in $H$ with vertices in both $S$ and $\overline{S} = V \backslash S$.  
When $H$ is a $2$-graph, the edges in $ \partial S$ have exactly one vertex in $S$ and one vertex in $\overline{S}$.

\begin{definition}
For a hypergraph $H$, the diameter $\text{diam}(H)$ is defined as the maximum distance between any pair of vertices. 
$$\text{diam}(H) = \max \{l(u,v): u,v \in V \}, $$
\end{definition}
\noindent where $l(u,v)$ is the length of the shortest path connecting $u$ and $v$.

The concept of algebraic connectivity, defined as the second smallest eigenvalue of the Laplacian (denoted by $\lambda_2 $) was introduced by 
Fiedler \cite{fiedler1973algebraic}. The algebraic connectivity has proved to be a reliable measure to understand the structure of a $2$-graph. 
Some of the results which give bounds on $\lambda_2 $ with respect to various graph constants are given below, including a version of the famous 
Cheeger inequality.

\begin{lemma}
\label{Fiedler} 
For a graph $G$ with $n$ vertices, 
$$ \lambda_2(G) = 2n \min \left\{ \frac{\sum_{(v_i,v_j) \in E }^{} (x_i - x_j)^2}
{\sum_{i=1}^{n} \sum_{j=1}^{n} (x_i - x_j)^2} : x \neq c. \textbf{1}_n, \text{ for all }c \in \R \right\}.$$
\end{lemma}

\begin{theorem}
\label{diamforgraph}
For a graph $G$ with $n$ vertices, 
$$ \lambda_2(G)\geq \frac{4}{n \cdot \text{diam}(G)} .$$
\end{theorem}

\begin{theorem}
Let $G$ be a 2-graph with more than one edge and $(d(v_1), \ldots, d(v_n))$ be the degree sequence of $G$. Then
\[  \lambda_2(G) \leq \min_{\{v_i ,v_j \} \in E} \left\{\frac{d(v_i) + d(v_j)-2}{2}\right\}. \]
\end{theorem}

\begin{theorem}
For a graph $G$,
$$2i(G) \geq \lambda_2(G) \geq \Delta(G) - \sqrt{\Delta(G)^2 -i(G)^2}.$$
\end{theorem}

For $k$-graphs the concept of analytic connectivity was introduced by Qi \cite{qi2014hvaluesoflaplacian} as the equivalent of $\lambda_2$. 
The definition is also valid for non-uniform hypergraphs as mentioned in \cite{banerjee2016spectra}.

\begin{definition}
The analytic connectivity of a $k$-uniform hypergraph $H$ is defined as 
$$\alpha(H) = \min_{j=1, \ldots, n} \min \{ \mathcal{L}x^m\ : x \in \R^n_+, \sum_{i=1}^{n} x_i^m =1, x_j =0 \}. $$
\end{definition}

The above results for $2$-graphs have been extended to $k$-graphs in \cite{li2016analytic} where the inequalities are presented with 
respect to $\alpha(H)$. The theorems are given below. 

\begin{theorem}
\label{diamforuniformgraph}
Let $H$ be a $k$-graph. Then 
$$ \alpha(H) \geq \frac{4}{n^2(k-1)\text{diam}(H) }.$$
\end{theorem}

\begin{theorem}
\label{degseqforuniformgraph}
Let $H$ be $k$-graph with more than one edge. Then 
$$ \alpha(H) \leq \min \left\{  \frac{d(v_{i_1}) + d(v_{i_2})+ \ldots +d(v_{i_k}) - k }{k} : \{v_{i_1}, v_{i_2}, \ldots, v_{i_k} \} \in E \right\}.$$
\end{theorem}

\begin{theorem}
\label{cheegerforuniformgraph} 
For a $k$-graph $H$ with $k\geq 3 $,
$$\frac{k}{2}i \geq \alpha(H) \geq \Delta - \sqrt{\Delta^2 -i^2}.$$
\end{theorem}

We also mention a lemma from \cite{li2016analytic} which will be used to prove the main results in the next section.

\begin{lemma}
\label{amgm2} 
Let $a=(a_1, \ldots, a_n) \in \R_{+}^{n}$. Let $A = (a_1+ \ldots+ a_n)/n $ and $G= (a_1\ldots a_n)^{1/n} $.
Then 
\begin{align}  
A-G &\geq \frac{1}{n(n-1)} \sum_{1\leq i <j \leq n}^{} (\sqrt{a_i}-\sqrt{a_j})^2. \\
A-G &\geq \frac{1}{n} \sum_{j=1}^{\lfloor n/2 \rfloor} (\sqrt{b_{j}}-\sqrt{b_{n+1-j}})^2.  
\end{align}
where $b_j = a_{\sigma (j)}$, for $j = 1,\ldots, n$ and $\sigma$ is a permutation of the set $\{1, \ldots,n \}$.
\end{lemma}

\section{Results for general hypergraphs}
\label{mainres}
In this section we prove the corresponding results for non-uniform hypergraphs. The inequalities are similar to that of 
uniform hypergraphs, except for an additional factor of $\frac{s_{\text{min}}}{m}$. The proofs can be obtained by 
modifying those in \cite{li2016analytic}. We have given the detailed proofs here for the sake of clarity and completeness.

\begin{theorem}
\label{diamforhypergraph}
Let $H$ be a general hypergraph. Then 
$$ \alpha(H) \geq \frac{ 4 s_{\text{min}}}{n^2m(m-1)\text{diam}(H)}.$$
\end{theorem}

\begin{proof}
Let $x=(x_1, x_2, \ldots,x_n) $ be the vector achieving $\alpha(H)$. Assume $x_n =0$ and let $y=x^{(m/2)}$.
Define a $2$-graph $H^*$ with vertex set $V(H)$ and $u \sim v$ in $H^*$ if and only if $\{u,v \} \subset e \in E$.
In other words $H^*$ is the clique expansion of $H$. We know that $\text{diam}(H) = \text{diam}(H^*)$.
For edge $e=\{v_{i_1}, \ldots, v_{i_s} \} \in E$ , consider $k_1$ copies of $x_1^m$, $k_2$ copies of $x_2^m $ 
and $k_s$ copies of $x_s^m$, where $k_j \geq 1 $ and $\sum k_j =m$. From Lemma \ref{amgm2} we have,
\begin{align*}
&\frac{k_1x_1^m + k_2x_2^m + \ldots k_sx_s^m }{m } - x_1^{k_1}x_2^{k_2}\ldots x_s^{k_s} \\
&\geq \frac{1}{m(m-1)} [ (x_1^{m/2} - x_2^{m/2})^2 + (x_1^{m/2} - x_3^{m/2})^2 + \ldots + 
(x_{s-1}^{m/2} - x_s^{m/2})^2] \\
&= \frac{1}{m(m-1)} \sum_{1\leq i <j \leq s} (x_{i}^{m/2} - x_j^{m/2})^2.
\end{align*}
Summing over different values of $k_1, k_2, \ldots , k_s$ we get
\begin{align*}
\frac{\Omega }{s} \left(\sum_{i=1}^{s} x_i^m \right) - x^e_m 
& \geq \frac{\Omega }{m(m-1)}  \sum_{1\leq i <j \leq s}(x_{i}^{m/2} - x_j^{m/2})^2 \\
\frac{1 }{s} \left(\sum_{i=1}^{s} x_i^m - \frac{s}{\Omega } x^e_m \right) 
& \geq \frac{1}{m(m-1)}  \sum_{1\leq i <j \leq s}(x_{i}^{m/2} - x_j^{m/2})^2 \\
\mathcal{L}(e)x^m & \geq \frac{s}{m}\frac{1}{(m-1)} \sum_{1\leq i <j \leq s}(x_{i}^{m/2} - x_j^{m/2})^2 \\
& \geq \frac{s_{\text{min}}}{m}\frac{1}{(m-1)} \sum_{1\leq i <j \leq s}(x_{i}^{m/2} - x_j^{m/2})^2.
\end{align*}
\begin{align*}
\alpha &= \sum_{e \in E(H)}  \mathcal{L}(e)x^m \\
&\geq \frac{s_{\text{min}}}{m}\frac{1}{(m-1)} \sum_{(v_iv_j) \in E(H^*)}(x_{i}^{m/2} - x_j^{m/2})^2 \\
&= \frac{s_{\text{min}}}{m}\frac{1}{(m-1)}\sum_{(v_iv_j) \in E(H^*)} (y_i -y_j)^2 \\
&= \frac{s_{\text{min}}}{m}\frac{1}{(m-1)} \sum_{i=1}^{n} \sum_{j=1}^{n}(y_i -y_j)^2
\frac{\sum_{(v_iv_j) \in E(H^*)} (y_i -y_j)^2}{\sum_{i=1}^{n} \sum_{j=1}^{n}(y_i -y_j)^2} \\
& \geq \frac{\lambda_2(H^*)}{2n(m-1)} \frac{s_{\text{min}}}{m}\sum_{i=1}^{n}\sum_{j=1}^{n}(y_i -y_j)^2 \quad
\text{ from Lemma } \ref{Fiedler}.
\end{align*}
\begin{align*}
 \sum_{i=1}^{n}\sum_{j=1}^{n}(y_i -y_j)^2 
 &= \sum_{i=1}^{n}\sum_{j=1}^{n}y_i^2 + \sum_{i=1}^{n}\sum_{j=1}^{n}y_j^2 - 
 2\sum_{i=1}^{n}\sum_{j=1}^{n}y_iy_j \\
 &= 2n \left(\sum_{i=1}^{n}y_i^2 \right) - 2\left(\sum_{i=1}^{n-1}y_i \right)^2  \quad \text{ (since $y_n =0$)} \\
 &\geq 2n -2(n-1)(\sum_{i=1}^{n}y_i^2) \quad \text{(from Cauchy-Schwarz) }\\
 &=2.
\end{align*}
We have $$ \alpha \geq \frac{\lambda_2(H^*)}{2n(m-1)} \frac{s_{\text{min}}}{m} \cdot 2. $$
From Theorem \ref{diamforgraph}, $ \lambda_2(H^*) \geq \frac{4}{\text{diam}(H^*) \cdot n}$. Therefore,
$$ \alpha \geq \frac{s_{\text{min}}}{m}\frac{4}{n^2 (m-1) \cdot \text{diam}(H^*)} .$$ 
\end{proof}

\begin{theorem}
\label{degseq}
  Let $H$ be a non-uniform hypergraph with more than one edge. Then 
  $$ \alpha(H) \leq \min \left\{  \frac{d(v_{i_1}) + d(v_{i_2})+ \ldots +d(v_{i_s}) - s }{s} : \{v_{i_1}, v_{i_2}, \ldots, v_{i_s} \} \in E \right\}.$$
\end{theorem}

\begin{proof} 
 Let $e_0 = \{ v_{i_1}, v_{i_2}, \ldots, v_{i_s} \} \in E$. Define a vector $x \in \R_{+}^{n}$ such that 
 \[
 x_i =
  \begin{cases}
   s^{-m} & \text{if } v_i \in e_0\\
   0 & \text{otherwise}.
  \end{cases}
\]
 Then $\displaystyle\sum\limits_{i=1}^{m} x_i^m = s(\frac{1}{s^{1/m}})^m =1$.
 $\mathcal{L}(e_0)x^m = \displaystyle\sum\limits_{i=1}^{m} x_i^m - \frac{s}{\Omega}x^e_m = 0. $
 \begin{align*}
 \alpha (H) &\leq \mathcal{L} x^m \\
	    &= \sum_{e \in E} \mathcal{L} (e)x^m \\
	    &=\left( \sum_{e \in E \backslash \{e_0 \}} \mathcal{L} (e)x^m + \mathcal{L} (e_0 )x^m  \right) \\ 
	    &= (d(v_{i_1})-1 )(1/s) + (d(v_{i_2})-1 )(1/s) + \ldots + (d(v_{i_s})-1 )(1/s) \\
	    &= \frac{d(v_{i_1}) + d(v_{i_2}) + \ldots + d(v_{i_s}) -s }{s}.
 \end{align*} 
\end{proof}

\begin{theorem} 
\label{cheegerfornonunigraph} 
For a non-uniform hypergraph $H$, 
$$ \frac{m}{2}i \geq \alpha(H) \geq \frac{s_{\text{min}}}{m}(\Delta - \sqrt{\Delta^2 -i^2}).$$
\end{theorem}

\begin{proof}
 Suppose $S \subset V$ gives the isoperimetric number $i$. Let $y=(y_1, \ldots, y_n ) \in \R^n_+$ be the vector 
 defined as follows.
 \[
 y_i =
  \begin{cases}
   \frac{1}{|S|^{1/m}} & \text{if } v_i \in S\\
   0 & \text{otherwise}.
  \end{cases}
\]
Let $t_S(e)= |\{v: v \in e \cap S \}|$ be the number of vertices of $e$ in $S$ and 
$t(S) = \frac{\sum_{e \in \partial S} t_S(e)}{|\partial S|}$. 
\begin{equation}
\label{mbound} 
t(S) + t(\overline{S}) = \frac{\sum_{e \in \partial S} t_S(e) +t_{\overline{S}}(e) }{|\partial S|} \leq m.
\end{equation}
\begin{equation}
\label{sumofal} 
\alpha \leq \mathcal{L}y^m = \left(\sum_{e \in S}+ \sum_{e \in \overline{S}}+ \sum_{e \in \partial S} \right)\mathcal{L}(e)y^m.
\end{equation}
If $e = \{v_1, \ldots, v_s \} \subset \overline{S}$, $\mathcal{L}(e)y^m = \sum y_i^m -\frac{s}{\Omega} y^e_m=0$. 
If $e = \{v_1, \ldots, v_s \} \subset {S}$,
\begin{align*}
 \mathcal{L}(e)y^m &= \sum_{i=1}^{m} y_i^m -\frac{s}{\Omega} y^e_m \\
 &= \frac{s}{|S|} - \frac{s}{\Omega}\sum \frac{1}{|S|}\cdot 1 \\
 &= 0.
\end{align*}
Therefore only edges in the boundary contribute to the sum of \eqref{sumofal}.
\begin{align*}
\alpha & \leq \sum_{e \in \partial S} \sum_{v_i \in e \cap S} y_i^m \\
      & = \sum_{e \in \partial S} \frac{t_S(e)}{|S|} = \frac{1}{|S|}t(S)|\partial S|\\
      & = t(S)\cdot i.
\end{align*}
Similarly we can get $\alpha \leq t(\overline{S})\cdot i$. Adding them we get $2\alpha \leq (t(S)+t(\overline{S}))i$.
Combining with ~\eqref{mbound}, we get $\alpha \leq \frac{mi}{2}$.

To prove the lower bound, suppose $x = (x_1, \ldots, x_n)$ achieves $\alpha$. For each edge $e=\{v_{i_1}, v_{i_2}, \ldots, v_{i_s} \}$ 
assume $x_{i_1} \leq x_{i_2} \leq \ldots \leq x_{i_s} $ by rearranging the vertices. We define a $2$-graph $\widehat{H}$ whose vertex set is same as 
$H$ and edges are such that $ E(\widehat{H}) = \cup_{e \in E(H)} \{v_{i_j} v_{i_{s+1-j}} : j=1, \ldots,\lfloor s/2 \rfloor\}$. Then 

\begin{equation}
\label{alphadef} 
\alpha = \sum_{e = \{v_{i_1}, v_{i_2}, \ldots, v_{i_s} \} } (\displaystyle\sum\limits_{j=1}^{s} x_{i_j}^m 
- \frac{s}{\Omega} \sum_{\substack{k_i \geq 1 \\ \sum k_i =m}} x_{i_1}^{k_1} \ldots x_{i_s}^{k_s} ).
\end{equation}
Consider $k_1$ copies of $x_1^m$, $k_2$ copies of $x_2^m $ and $k_s$ copies of $x_s^m$, where $k_j \geq 1 $ and $\sum k_j =m$. 
Applying Lemma \ref{amgm2} we get 
\begin{equation} 
\label{amgmalt}
 \frac{k_1x_1^m + k_2x_2^m + \ldots k_sx_s^m }{m } - x_1^{k_1}x_2^{k_2}\ldots x_s^{k_s} \\
\geq \frac{1}{m} \sum_{i=1}^{\lfloor m/2\rfloor}  (\sqrt{b_i} - \sqrt{b_{m+1-i}})^2
\end{equation}
where $b_1, \ldots, b_m$ is any permutation of the variables $x_i$. In particular consider the assignment  
$b_1= x_{i_1}^{m}, b_2 = x_{i_2}^{m}, \ldots, b_{\lfloor s/2 \rfloor} = x^{m}_{i_{s /2}}, 
b_{m+1 -\lfloor s/2\rfloor} = x_{i_{s+1 -\lfloor s/2 \rfloor}}^{m}, \ldots,b_{m} = x_{i_s}^{m} $, and 
$ b_{\lfloor s/2 \rfloor +1} , \ldots , b_{m-\lfloor s/2\rfloor}$ are assigned to any of the remaining 
$m-s $ variables. Then 
\begin{align*}
\sum_{i=1}^{\lfloor m/2\rfloor}  (\sqrt{b_i} - \sqrt{b_{m+1-i}})^2 
& \geq \sum_{i=1}^{\lfloor s/2\rfloor}  (\sqrt{b_i} - \sqrt{b_{m+1-i}})^2 \quad (\text{ since }s \leq m) \\
& = \sum_{j=1}^{\lfloor s/2\rfloor} (\sqrt{x^m_{i_j} }  - \sqrt{x^m_{i_{s+1-j}}} )^2
\end{align*}

Using the above inequality with \eqref{amgmalt}, and summing over all possible values of $k_1, \ldots, k_s$ as in 
proof of Theorem \ref{diamforhypergraph} we get
\begin{align*}
\alpha & \geq \sum_{e= \{v_{i_1}, \ldots, v_{i_s} \} \in E(H) } (\frac{s}{m} \sum_{j=1}^{\lfloor s/2\rfloor} (\sqrt{x^m_{i_j} }  - \sqrt{x^m_{i_{s+1-j}}} )^2) \\
& \geq \frac{s_{\text{min}}}{m} \sum_{e= \{v_{i_1}, \ldots, v_{i_s} \} \in E(H) } (\sum_{j=1}^{\lfloor s/2\rfloor} (\sqrt{x^m_{i_j} }  - \sqrt{x^m_{i_{s+1-j}}} )^2) \\
& = \frac{s_{\text{min}} }{m} \sum_{\{v_{i},v_{j} \} \in E(\widehat{H}) } (\sqrt{x^m_{i} }  - \sqrt{x^m_{j}} )^2).
\end{align*}
Proceeding as in proof of $k$-graphs in \cite{li2016analytic}, let $M =\sum_{\{v_{i},v_{j} \} \in E' } (y_i - y_j)^2$ where $E'=E(\widehat{H}) $ and 
$y_i = \sqrt{x_i^m} $. From Cauchy-Schwarz inequality we have 
\begin{equation}
\label{expforM} 
M \geq \frac{(\sum_{\{v_{i},v_{j} \} \in E' }|y_i^2 - y_j^2|)^2}{\sum_{\{v_{i},v_{j} \} \in E' } (y_i + y_j)^2} .
\end{equation}

Let $w_0 (=0) < w_1 < \ldots < w_h $ be the distinct values of $y_i$, for $i=1, \ldots, n $. For $j=0, \ldots, h$, 
let $V_j = \{ v_i \in V : y_i \geq w_j\}$. For each edge $e \in \partial(V_j)$, 
let $\delta_j(e) = \min \{|V_j \cap e|, |\overline{V_j} \cap e|\}$. Let $\delta(V_j) = \min \{ \delta_j(e) : e \in \partial(V_j)\}$ and 
$\delta(H) = \min_{j =0, \ldots, h } \delta(V_j)$.

\begin{align*}
\sum_{\{v_{i},v_{j} \} \in E' }|y_i^2 - y_j^2|
&= \sum_{k=1}^{h}\sum_{\substack{ \{v_{i},v_{j} \} \in E' \\ v_i \in V_k \\ v_j \notin V_k} }(y_i^2 - y_j^2) = 
 \sum_{k=1}^{h}\sum_{\substack{ \{v_{i},v_{j} \} \in E' \\ y_i = w_k \\ y_j = w_l \\ l<k} }(w_k^2 - w_l^2)  \\
&= \sum_{k=1}^{h}\sum_{\substack{ \{v_{i},v_{j} \} \in E' \\ y_i = w_k \\ y_j = w_l \\ l<k} }
(w_k^2 - w_{k-1}^2) + (w_{k-1}^2 - w_{k-2}^2)+ \ldots+ (w_{l+1}^2 - w_l^2) \\
&= \sum_{k=1}^{h} \sum_{\substack{\{v_{i},v_{j} \} \in E' \\ v_i \in V_k } }\sum_{v_j \notin V_k} (w_k^2 - w_{k-1}^2) \\
&\geq \sum_{k=1}^{h} \delta(V_k)|\partial V_k|(w_k^2 - w_{k-1}^2) \\
&\geq \sum_{k=1}^{h} \delta(H)i(H)|V_k|(w_k^2 - w_{k-1}^2) \\
&= \delta(H)i(H) (|V_{h}| (w_{h}^2-w_{h-1}^2) + \ldots +|V_{1}| (w_{1}^2-w_{0}^2)   ) \\
&= \delta(H)i(H) ((|V_{h}| -|V_{h-1}|)w_{h}^2 + \ldots +(|V_{1}| -|V_{2}|)w_{1}^2  ) \\
&= \delta(H)i(H) \sum_{i=1}^{n} y_i^2
\end{align*}

\begin{align*}
\sum_{\{v_{i},v_{j} \} \in E' } (y_i + y_j)^2 
&= 2 \sum_{v_i, v_j \in E'} (y_i^2 + y_j^2) - \sum_{v_i, v_j \in E'}(y_i - y_j)^2 \\
& \leq 2 \sum_{i=1}^{n} d(v_i) y_i^2 - \sum_{v_i, v_j \in E'}(y_i - y_j)^2 \\          
& \leq 2 \Delta (\widehat{H}) \sum_{i=1}^{n} y_i^2 - \sum_{v_i, v_j \in E'}(y_i - y_j)^2 \\
& = 2  \Delta (\widehat{H}) -M \leq  2\Delta (H) -M.
\end{align*}

$$M \geq \frac{\delta(H)^2 i(H)^2}{2\Delta -M} \geq \frac{i(H)^2}{2\Delta -M}, $$
solving which we get $M \geq \Delta - \sqrt{\Delta^2 -i^2}$. Substituting in \eqref{expforM} we get
$$\alpha \geq \frac{s_{\text{min}}}{m}(\Delta - \sqrt{\Delta^2 -i^2}). $$

\end{proof}

\section{Summary}
\label{summary}

In this paper we have built upon the results for analytic connectivity of $k$-graphs given in \cite{li2016analytic} by applying the definitions for general hypergraphs found in 
\cite{banerjee2016spectra}. We have obtained bounds for the analytic connectivity of general hypergraphs with respect to diameter and degree sequence of the hypergraph. We also proved a version of Cheeger's inequality for hypergraphs.


\begin{thebibliography}{10}
\bibitem{banerjee2016spectra}
Anirban Banerjee, Arnab Char, and Bibhash Mondal.
\newblock Spectra of general hypergraphs.
\newblock {\em arXiv preprint:1601.02136}, 2016.

\bibitem{berge1973graphs}
Claude Berge and Edward Minieka.
\newblock {\em Graphs and hypergraphs}, volume~7.
\newblock North-Holland publishing company Amsterdam, 1973.

\bibitem{bu2016spectral}
Changjiang Bu, Jiang Zhou, and Lizhu Sun.
\newblock Spectral properties of general hypergraphs.
\newblock {\em arXiv preprint:1605.05942}, 2016.

\bibitem{cooper2012spectra}
Joshua Cooper and Aaron Dutle.
\newblock Spectra of uniform hypergraphs.
\newblock {\em Linear Algebra and its Applications}, 436(9):3268--3292, 2012.

\bibitem{fiedler1973algebraic}
Miroslav Fiedler.
\newblock Algebraic connectivity of graphs.
\newblock {\em Czechoslovak mathematical journal}, 23(2):298--305, 1973.

\bibitem{li2016analytic}
Wei Li, Joshua Cooper, and An~Chang.
\newblock Analytic connectivity of k-uniform hypergraphs.
\newblock {\em Linear and Multilinear Algebra}, pages 1--13, 2016.

\bibitem{lim2006}
L.H. Lim.
\newblock Singular values and eigenvalues of tensors, a variational approach.
\newblock {\em Proceedings of 1st IEEE international workshop on computational
  advances of multitensor adaptive processing}, pages 129--132, 2005.

\bibitem{qi2005eigenvalues}
Liqun Qi.
\newblock Eigenvalues of a real supersymmetric tensor.
\newblock {\em Journal of Symbolic Computation}, 40(6):1302--1324, 2005.

\bibitem{qi2014hvaluesoflaplacian}
Liqun Qi.
\newblock $H^+$ eigenvalues of laplacian and signless laplacian tensor.
\newblock {\em Communications in Mathematical Sciences}, 12(6):1045--1064,
  2014.
\end{thebibliography}

\end{document}